\documentclass[manuscript]{acmart}

\AtBeginDocument{%
  }

\setcopyright{acmcopyright}
\copyrightyear{2022}
\acmYear{2022}
\acmDOI{XXXXXXX.XXXXXXX}
\usepackage{xcolor,pgf,tikz}
\usetikzlibrary{arrows}

\newcommand{\mcsg}{\mathit{mcsg}}

\newcommand{\sfh}{{\sf h}}
 \newcommand{\dom}{\mathrm{Dom}}
 \newcommand{\ran}{\mathrm{Ran}}
 \newcommand{\pos}{\mathit{pos}}
\newtheorem{observation}{Observation}
\newtheorem{convention}{Convention}

\begin{document}

\title[One or Nothing]{One or Nothing: Anti-unification over the Simply-Typed Lambda Calculus}
\titlenote{Funded by  Czech Science Foundation Grant No. 22-06414L and Math$_{LP}$ project (LIT-2019-7-YOU-213) of the  Linz Institute of Technology and the state of Upper Austria}
\author{David M. Cerna}
\authornotemark[2]
\email{dcerna@cs.cas.cz}
\orcid{0000-0002-6352-603X}
\affiliation{%
  \institution{Czech Academy of Sciences Institute of Computer Science (CAS ICS)}
  \city{Prague}
  \country{Czechia}}
\affiliation{%
  \institution{Research Institute for Symbolic Computation (RISC), JKU}
  \city{Linz}
  \country{Austria}
}
\author{Michal Buran}
\email{michal.buraan@gmail.com}
\orcid{0000-0002-1196-3799}
\affiliation{%
  \institution{Research Institute for Symbolic Computation (RISC), JKU}
  \city{Linz}
  \country{Austria}
}

\renewcommand{\shortauthors}{ D.M. Cerna \& M. Buran}

\begin{abstract}
Generalization techniques have many applications, including template construction, argument generalization, and indexing. Modern interactive provers can exploit advancement in generalization methods over expressive type theories to further develop proof generalization techniques and other transformations. So far, investigations concerned with anti-unification (AU) over $\lambda$-terms and similar type theories have focused on developing algorithms for well-studied variants. These variants forbid the nesting of generalization variables, restrict the structure of their arguments, and are \textit{unitary}. Extending these methods to more expressive variants is important to applications. We consider the case of nested generalization variables and show that the AU problem is \textit{nullary} (using \textit{capture-avoiding} substitutions), even when the arguments to free variables are severely restricted.  
\end{abstract}



\maketitle

\section{Introduction}
{\em Anti-unification} (generalization)~\cite{CernaK23}, first introduced in the 1970s by Plotkin and Reynolds~\cite{Plotkin70,Reynolds70}, is a process that derives from a set of symbolic expressions a new symbolic expression possessing certain commonalities shared between its members. For example, $f(x)$ generalizes the terms $f(a)$ and $f(b)$.  The concept has been applied to inductive theorem proving based on tree grammars~\cite{EberhardH15}, library compression~\cite{DBLP:journals/pacmpl/CaoKNWTP23} recursion scheme detection in functional programs~\cite{DBLP:journals/fgcs/BarwellBH18}, inductive synthesis of recursive functions~\cite{DBLP:books/sp/Schmid03}, learning fixes from software code repositories~\cite{DBLP:journals/corr/abs-1803-03806,DBLP:journals/pacmpl/BaderSP019}, preventing bugs and misconfiguration~\cite{246332}, as well as uses within the fields of natural language processing and linguistics~\cite{DBLP:books/sp/Galitsky19,EasyChair:203}. Applications, most related to our work are proof generalization~\cite{pfenning1991unification} and higher-order term indexing~\cite{10.1145/1614431.1614437}, which use anti-unification over typed $\lambda$-terms and their extensions~\cite{DBLP:books/daglib/0032840}. 

 A relatively recent work~\cite{DBLP:conf/rta/CernaK19} provided a general framework for anti-unification over typed $\lambda$-terms where the authors restricted the structure of the generalizations to so-called top-maximal shallow $\lambda$-terms. Shallow means that free variable occurrences cannot occur as an argument to other free variable occurrences, i.e., the term $\lambda y.X(y)$ is shallow while $\lambda y.X(X(y))$ is not. Top maximality is difficult to formally describe but easy to illustrate by example. Consider the terms $s=\lambda x.f(f'(x))$ and $t=\lambda x.f(h(x))$. The terms $\lambda x.Y(f(f'(x)),f(h(x)))$ and  $\lambda x.f(Y(f'(x),h(x)))$ are both generalizations of $s$ and $t$, but only the latter term explicitly captures that $s$ and $t$ have $f$ at the same position. Essentially, top-maximality means free variables generalize the differences closest to the root of the term tree and avoid generalizing similarities at positions closer to the root. Important for this work, the pattern generalization of two terms is always top-maximal shallow as constants, and $\lambda$ abstractions cannot occur as arguments to the free variables. 

In~\cite{DBLP:conf/rta/CernaK19}, the authors considered variants of anti-unification over $\lambda$-terms based on the concept of \textit{restricted term} introduced for \textit{function-as-constructor unification}~\cite{DBLP:journals/amai/LibalM22}. Restricted terms allow a generalization of the pattern restriction without compromising the unitarity of unification, i.e., there is a single \textit{most general unifier}. In addition, function-as-constructor unification requires certain relations between the arguments of free variable occurrences and between free variable occurrences to hold. Relaxing these restrictions allows one to define a few additional variants. Similarly,  anti-unification in this setting is \textit{unitary}: a unique \textit{least general generalization} exists.

Furthermore, anti-unification remains unitary even when one relaxes most of the restrictions introduced in~\cite{DBLP:journals/amai/LibalM22}. Though these results are significant, shallow AU is quite restrictive for applications such as proof generalization and templating within modern \textit{interactive theorem provers}~\cite{bertot2013interactive,MouraKADR15}. In this paper, we consider deep AU over the simply-typed $\lambda$-calculus and show that an elegant framework, as presented in earlier work concerning top-maximal shallow generalization, cannot exist because the theory is nullary (least general generalizations do not exist), even when we severely restrict arguments to the free variables; this entails that, for applications such as proof generalization, templating, and augmentation, rather than extending the existing framework into the deep setting, future investigations will have to consider extensions into more expressive type theories and computational frameworks.  

\section{Preliminaries}
\label{prelims}
 Higher-order signatures are composed of \emph{types} (denoted by $\mathcal{T}$) constructed from a set
of \emph{base types} (typically $\mathcal{B}$) using the grammar $\tau ::= \mathcal{B}\ \mid\ \tau \rightarrow \tau$. When it is possible to avoid confusion, we will abbreviate types of the form $\gamma_1\rightarrow\cdots \rightarrow \gamma_m\rightarrow \alpha$ as $\overline{\gamma_m}\rightarrow \alpha$. \emph{Variables} are separated into sets: \textit{free} variables $\mathcal{FV}$ (denoted
$X, Y, Z,  \ldots$) and \textit{bound} variables $\mathcal{BV}$ (denoted $x, y, z,\ldots$). While it is non-standard to distinguish free and bound variables, doing so simplifies the argument presented below. A similar approach was taken in~\cite{DBLP:conf/rta/CernaK19}, where terms were discussed within a variable context. \emph{Constants} are denoted $f, a, b, c, \ldots$, the set of all constants is denoted $\mathcal{C}$ and $\Sigma\subseteq \mathcal{C}$ is a \textit{signature}. Free variables, bound variables, and constants are assigned types from $\mathcal{T}$. By the symbol $\sfh$, we denote a constant or a variable.

\emph{$\lambda$-terms} (typically denoted $t,s,u,\ldots$) over a signature $\Sigma$ are constructed using the grammar  $t ::= x:\tau\mid Y:\tau \mid c:\tau \mid \lambda x:\tau.t \mid t_1\ t_2$
where $x$ is a bound variable of type $\tau\in \mathcal{T}$, $Y$ is a free variable of type $\tau\in \mathcal{T}$, and $c$ is a constant contained in $\Sigma$ of type  $\tau\in \mathcal{T}$. Note that neither constants nor variables (bound or free) can be assigned multiple types. When clear from context, we will drop type annotations. A $\lambda$-term is considered \textit{closed} when it does not contain free variables. For example, $\lambda x.\lambda y.x$ is closed and  $\lambda x.\lambda y.Z(x,y)$ is not as $Z$ is a free variable. The set of all $\lambda$-terms over the signature $\Sigma$ is denoted by $\Lambda(\Sigma)$ and the set of closed $\lambda$-terms  over the signature $\Sigma$ is denoted by $\Lambda_c(\Sigma)\subset \Lambda(\Sigma)$. When the signature is clear from context, we write $\Lambda$ ($\Lambda_c$ for closed terms). 

We assume that for all terms in $\Lambda(\Sigma)$, bound variables do not occur unabstracted and are in $\eta$-long $\beta$-normal form unless otherwise stated. That is, in addition to  $\beta$-normal form,  constants and free and bound variables are fully abstracted, i.e., $f:\gamma_1\rightarrow \gamma_2$ is written as $\lambda x.f(x)$. Given a lambda term $t\in \Lambda$, $t\downarrow_{\eta}$ denotes the $\eta$-normal form of $t$. For example, let $f:\alpha\rightarrow \alpha$ be a constant, $x:\alpha\rightarrow \alpha$ a bound variable,  then $t= \lambda x. f(\lambda y.x(y))$ is a $\lambda$-term in $\eta$-long $\beta$-normal form. The  $\eta$-normal form of $t$ is $t\downarrow_{\eta} = \lambda x. f(x)$.

Terms of the form $(\ldots (\sfh\ t_m ) \ldots t_1 )$ 
are  written as $\sfh(t_1 , \ldots , t_m)$, and terms of the form $\lambda x_1 . \ldots\lambda x_n .t$ as $\lambda x_1 , \ldots , x_n .t$. We use $\overline{x_n}$ as a short-hand for $x_1 , \ldots , x_n$. We define $\mathit{head}(t)$ for $t\in \Lambda(\Sigma)$ as $\mathit{head}(\lambda x.t') = \lambda x.$ and $\mathit{head}(\sfh(\overline{s_m})) = \sfh$. By $\mathit{occ}(\sfh, t)$ for $t\in \Lambda(\Sigma)$ we denote the number of occurrences of $\sfh$ in $t$, and by $\mathcal{FV}(t)$ (resp. $\mathcal{BV}(t)$), we denote the set of free (bound) variables occurring in $t$.

The set of \emph{positions} of a term $t$, denoted by $\pos(t)$, is the set of strings of positive integers, defined as $\pos(\lambda x.t)=\{\epsilon \}\cup \{1. p \mid  p\in \pos(t)\}$ and $\pos(\sfh(t_1,\ldots,t_n))=\{\epsilon \}\cup  \bigcup_{i=1}^n \{i. p \vert  p\in \pos(t_i)\}$, where $\epsilon$ denotes the empty string. For example, the term at position $1.1.2$ of 
$\lambda x.f(x,a)$ (i.g., $\lambda x.((f\ a)\ x)$) is $a$.

 If $p\in \mathit{pos}(s)$ and $s$ is a term, then $s|_p$ denotes the subterm of $s$ at position $p$.  For positions $p$ and $q$, we say $p<q$ if $q=p.q'$ for some nonempty string $q'$.

\emph{Substitutions} are finite sets of pairs $\{X_1\mapsto
t_1,\ldots, X_n\mapsto t_n\}$ where $X_i$ and $t_i$  have the same type, $X_i\not = t_i$, and the $X$'s are pairwise distinct free variables. Substitutions represent functions that map each $X_i$ to $t_i$ and any other variable to itself. They can be extended to
type preserving functions from terms to terms and, as usual, avoid
variable capture (variables introduced by substitutions are not captured by $\lambda$s occurring in the term). We use postfix notation for substitution applications, writing $t\sigma$ instead of $\sigma(t)$. The  \emph{domain} of a substitution $\sigma$, denoted $\dom(\sigma)$, is the set $\{ X\ \vert \ X\sigma \not = X\}$. The \emph{range} of a substitution $\sigma$, denoted $\ran(\sigma)$, is the set $\{ t\ \vert \ X\in \dom(\sigma)\ \wedge \ X\sigma = t\}$. Substitutions are denoted by lowercase Greek letters. As usual, the application $t\sigma$
affects only the free variable occurrences of $t$ whose free variable is found in $\dom(\sigma)$.

We consider two variants of anti-unification over the simply-typed lambda calculus, $\Lambda(\Sigma)$ and $\Lambda_{sp}(\Sigma)\subset \Lambda(\Sigma)$ (See Definition~\ref{def:superpattern}). Below, $\mathcal{L}\in 
\{\Lambda(\Sigma),\Lambda_{sp}(\Sigma)\}$. We will use the notation for $\mathcal{L}$-terms to specify a term from one of the two variants. A term $g\in \Lambda(\Sigma)$ is referred to as an $\mathcal{L}$-generalization of terms $s,t\in 
\Lambda_c(\Sigma)$, where $s$ and $t$ have the same type, if there exist substitutions $\sigma_1$ and $\sigma_2$ such that $g\sigma_1 
=_{\alpha\beta\eta} s$ and $g\sigma_2=_{\alpha\beta\eta}t$, $g$ is an $\mathcal{L}$-terms, 
where $=_{\alpha\beta\eta}$ means equivalent with respect to $\alpha$-equivalence, $\beta$-reduction, and $\eta$-expansion. When clear 
from context, we drop $\alpha\beta\eta$ from the equality relation. We refer to such $g$ as a solution to the 
\textit{$\mathcal{L}$-anti-unification problem ($\mathcal{L}$-AUP)} $s\triangleq_{\mathcal{L}} t$. Given a generalization $g$ of $s\triangleq_{\mathcal{L}} t$,  $\mathcal{GS}(s,t,g)$ denotes the set of pairs of substitutions $(\sigma_1,\sigma_2)$ such that $g\sigma_1= s$ and $g\sigma_2 = t$.  Furthermore, $\mathcal{GS}(s,t,g)_{gnd}\subset \mathcal{GS}(s,t,g)$ denotes the pairs of substitutions $(\sigma_1,\sigma_2)$ such that $\{X\vert X\in \mathcal{FV}(t) \wedge t\in \mathit{Ran}(\sigma_1)\cup \mathit{Ran}(\sigma_2)\}=\emptyset$.

We only consider the anti-unification problems between closed  $\lambda$-terms as generalization of general $\lambda$-terms can be reduced to generalization of closed $\lambda$-terms. For example, the free variables of the input terms can be replaced by constants and then reintroduced after the computation of the generalization; this is common practice in anti-unification literature and was also assumed by the previous investigations~\cite{DBLP:conf/rta/CernaK19}. 

The set of all  $\mathcal{L}$-generalizations of $s,t \in \Lambda_c(\Sigma)$ is denoted by $\mathcal{G}_{\mathcal{L}}(s,t)$. A quasi-ordering 
$\leq_{\mathcal{L}}$ may be defined on $\mathcal{L}$ as follows: for $\mathcal{L}$-terms $g,g'$, $g\leq_{\mathcal{L}} g'$ if and only if 
there exists $\sigma$ such that $g\sigma =_{\alpha\beta\eta} g'$. By $g <_{\mathcal{L}} g'$, we denote $g \leq_{\mathcal{L}} g'$ and 
$g'\not \leq_{\mathcal{L}} g$. A minimal complete set of $\mathcal{L}$-generalizations of $s$ and $t$, denoted $\mcsg_{\mathcal{L}}(s,t)$  is a set of terms with the following properties: 
\begin{enumerate}
\item $\mcsg_{\mathcal{L}}(s,t)\subseteq \mathcal{G}_{\mathcal{L}}(s,t)$.
\item For all $g'\in \mathcal{G}_{\mathcal{L}}(s,t)$ there exists a $g\in \mcsg_{\mathcal{L}}(s,t)$ such that $g' \leq_{\mathcal{L}} g$.
\item If $g,g'\in \mcsg_{\mathcal{L}}(s,t)$ and $g\not = g'$,then  $g  \not \leq_{\mathcal{L}} g'$ and $g'  \not \leq_{\mathcal{L}} g$.
\end{enumerate}
A subset of $\mathcal{G}_{\mathcal{L}}(s, t)$ is called {\em complete} if the first two conditions hold. In anti-unification literature~\cite{CernaK23}, theories are classified by the size and existence of the $\mcsg_{\mathcal{L}}(s,t)$: 
\begin{itemize}
\item \textit{Unitary}: The $\mcsg_{\mathcal{L}}(s,t)$ exists for all  $s,t\in  \Lambda_c(\Sigma)$ and is always singleton. 
\begin{itemize}
    \item[] \underline{Examples:} syntactic anti-unification (AU)~\cite{Plotkin70,Reynolds70}, top-maximal, shallow AU over $\lambda$-terms~\cite{DBLP:journals/jar/BaumgartnerKLV17,DBLP:conf/rta/CernaK19}, nominal AU with Atom variable~\cite{schmidtschau_et_al:LIPIcs.FSCD.2022.7} (depends on constraints).
\end{itemize}
\item \textit{Finitary}: The $\mcsg_{\mathcal{L}}(s,t)$ exists, is finite for all  $s,t\in  \Lambda_c(\Sigma)$, and there exist closed terms $s',t'\in  \Lambda_c(\Sigma)$  for which $1< |\mcsg_{\mathcal{L}}(s',t')| <\infty$. 
\begin{itemize}
    \item[] \underline{Examples:}  associative (A), commutative (C), and AC AU ~\cite{DBLP:journals/iandc/AlpuenteEEM14}, AU with a single unit equation~\cite{DBLP:journals/mscs/CernaK20}, unranked term-graphs AU~\cite{DBLP:conf/rta/BaumgartnerKLV18}, hedges~\cite{KutsiaLV14}.
\end{itemize}
\item \textit{Infinitary}: The $\mcsg_{\mathcal{L}}(s,t)$ exists for all  $s,t\in  \Lambda_c(\Sigma)$, and there exists at least one pair of closed terms for which it is infinite. 
\begin{itemize}
      \item[] \underline{Examples:}  Idempotent AU~\cite{ourpaper} and Absorption AU~\cite{ayalarincon2023equational}.
\end{itemize}
\item \textit{Nullary}: There exist  terms $s,t\in  \Lambda_c(\Sigma)$ such that $\mcsg_{\mathcal{L}}(s,t)$ does not exist. 
\begin{itemize}
      \item[] \underline{Examples:} AU with multiple unit equations~\cite{FSCD2020}, semirings and idempotent-absorption~\cite{DBLP:journals/tcs/Cerna20}.
\end{itemize}
\end{itemize}

\section{Results}
We assume a set of base types $\mathcal{B}$ that contains at least one base type, which we denote by $\alpha$. Furthermore, for 
every type $\tau\in\mathcal{T}$ constructible using $\mathcal{B}$, there is a constant symbol $c\in \Sigma$ of type $\tau$, where 
$\Sigma$ is the signature over which our problem is defined. This assumption is important for the transformation of 
$\mathcal{L}$-pattern-derived generalizations (defined below) into $\mathcal{L}$-tight generalizations 
(Definition~\ref{def:patternDerived}) presented in Lemma~\ref{lem:minimal}. This transformation requires the application of 
substitutions whose range may include variables of type $\tau$ such that the chosen signature does not contain a constant of type 
$\tau$. To avoid this situation, we assume $\Sigma$ contains at least one constant of every type in $\mathcal{T}$, allowing us to 
assume that the substitutions considered in the proof of Lemma~\ref{lem:minimal} do not introduce free variables. For the rest of 
this paper, $\Sigma$ will be the signature defined above, and all discussed $\lambda$-terms will be constructed using this 
signature. In addition to general $\lambda$-terms over $\Sigma$, we also consider so-called \textit{superpatterns}:

\begin{definition}[Superpatterns]
\label{def:superpattern}
A $\lambda$-term $t\in \Lambda$ is a \emph{superpattern} if, for every $p\in \mathit{pos}(t)$ such that  $t\vert_{p}\downarrow_{\eta} = X(\overline{s_m})$ and $X\in \mathcal{FV}$, the following conditions hold: 
\begin{itemize}
    \item For all $1\leq i\leq m$, $s_i\downarrow_{\eta}\in \mathcal{BV}(t)$, or $\mathit{head}(s_i\downarrow_{\eta})\in \mathcal{FV}$.
    \item For all $1\leq i<j\leq m$, such that  $s_i\downarrow_{\eta},s_j\downarrow_{\eta}\in \mathcal{BV}(t)$, $s_i\downarrow_{\eta} = s_j\downarrow_{\eta}$ iff $i=j$.
\end{itemize}
We denote the set of superpatterns by $\Lambda_{sp}$.
\end{definition}
Essentially, \textit{superpatterns} are a deep extension of the \textit{pattern} fragment. We consider generalization over the \textit{superpattern} fragment as it is a minimal deep extension of shallow fragments considered in~\cite{DBLP:conf/rta/CernaK19}.
\begin{example}
Examples of superpatterns are 
\begin{itemize}
    \item $\lambda x.\lambda y.Z$, $\lambda x.\lambda y.Z(x,y)$
    \item $\lambda x.\lambda y.Z(R(x,y),K(x),K(P(x)),W,y)$
\end{itemize} 
While the following are not superpatterns
\begin{itemize}
    \item $\lambda x.\lambda y.Z(x,x)$, $\lambda x.\lambda y.Z(\lambda w.w(x),\lambda w.w(y))$
    \item $\lambda x.\lambda y.Z(R(x,y),y,K(x),y,K(P(x)),y,W)$
    \item $\lambda x. \lambda y.Z (y(\lambda u.u))$
\end{itemize}
\end{example}

Let $f\in \Sigma$ be a constant of type $\alpha\rightarrow\alpha$. We use the following anti-unification problem to show that anti-unification over both $\Lambda$ and $\Lambda_{sp}$ is nullary, where $\mathcal{L}\in \{\Lambda,\Lambda_{sp}\}$:
$$\lambda x:\alpha.\lambda y:\alpha. f(x)\triangleq_{\mathcal{L}} \lambda x:\alpha.\lambda y:\alpha. f(y).$$ 
  
\begin{convention}
    The terms $\lambda x:\alpha.\lambda y:\alpha. f(x)$ and  $\lambda x:\alpha.\lambda y:\alpha. f(y)$ will be abbreviated as $s$ and  $t$ respectively.
\end{convention}
 Using the algorithm presented in~\cite{DBLP:journals/jar/BaumgartnerKLV17,DBLP:conf/rta/CernaK19}, the least general pattern generalization of $s\triangleq_{\mathcal{L}}t$ is $\lambda x:\alpha.\lambda y:\alpha. f(Z(x,y))$.  Note that  $(\{Z\mapsto\lambda x, y. x\},\{Z\mapsto\lambda x, y. y\})\in \mathcal{GS}_{gnd}(s,t,\lambda x, y. f(Z(x,y)))$, i.e., the left and right projections are used to recover the original terms. Given that we are considering a more expressive language than the previous work, generalizations exist of $s\triangleq_{\mathcal{L}}t$ more specific than $\lambda x, y. f(Z(x,y))$. For example, 
$\lambda x, y. f(W(W(x,y),W(x,y)))$. We aim to show that any complete set of generalizations will contain a generalization $g$ at least as specific as the pattern generalization and that there exists a $g'$ more specific than $g$. We refer to generalizations $g$ such that $\lambda x.\lambda y.f(Z(x,y))\leq_{\mathcal{L}}g$ as \textit{$\mathcal{L}$-pattern-derived}.

\begin{lemma}
\label{lem:uppermostfree}
Let  $\mathcal{L}\in\{\Lambda,\Lambda_{sp}\}$ and $g= \lambda x.\lambda y.f(r)$ be an  \textit{$\mathcal{L}$-pattern-derived} generalization of $s\triangleq_{\mathcal{L}}t$. Then there exist $Z\in \mathcal{FV}(g)$ of type $\overline{\gamma_m}\rightarrow \alpha$ such that $\mathit{head}(r\downarrow_{\eta})= Z$  where  $m>0$ and for all $1\leq i\leq m$, $\gamma_i\in \mathcal{T}$. 
\end{lemma}
\begin{proof}
Let us consider the possible cases for $\mathit{head}(r\downarrow_{\eta})$: 
\begin{itemize}
    \item If $\mathit{head}(r\downarrow_{\eta}) = c \in \Sigma$ of type $\overline{\beta_n}\rightarrow \alpha$ then $g= \lambda x.\lambda y.f(c(r_1,\cdots r_n))$ where for all $1\leq i\leq n$, $r_i$ is of type $\beta_i$. Note that there do not exist substitutions $\sigma_1$ and $\sigma_2$ such that $g\sigma_1 = s$ and $g\sigma_2= t$ because neither term contains an application of $f$ to a constant $c$. Thus, $g$ is not a generalization of $s\triangleq_{\mathcal{L}}t$.
    \item  If $\mathit{head}(r\downarrow_{\eta})\in \mathcal{BV}(g)$, then it must be either $x$ or $y$ as no other bound variables are abstracted by a $\lambda$ above the position of $r$ in $g$. However, $s$ and $t$ disagree concerning which of the two bound variables $f$ takes as an argument. Thus, $g$ cannot be a generalization of $s\triangleq_{\mathcal{L}}t$. 
    \item If $\mathit{head}(r\downarrow_{\eta}) = \lambda z.$ where $z$ is of type $\gamma$,  then $r$ is at least of type $\gamma\rightarrow \alpha$. We assume our generalizations are well-typed and in $\eta$-long $\beta$-normal form, thus implying a type mismatch between $f$ and $r$. Thus, $g$ cannot be a generalization of $s\triangleq_{\mathcal{L}}t$ as we assume all terms are in $\eta$-long $\beta$-normal form and thus $r$ would be of the incorrect type to be an argument of $f$. 
\end{itemize}
The only possibility not yet considered is $\mathit{head}(r\downarrow_{\eta})=Z\in \mathcal{FV}(g)$ of type $\overline{\gamma_m}\rightarrow \alpha$ where  $m>0$. Note that if $m=0$, then $Z$ cannot be instantiated by the bound variables $x$ and $y$ (remember, capture avoiding substitutions) and thus would not be a generalization of $s\triangleq_{\mathcal{L}}t$.
\end{proof}
\noindent Using Lemma~\ref{lem:uppermostfree} we can deduce that $\mathcal{L}$-pattern-derived generalizations of  $s\triangleq_{\mathcal{L}}t$ must be of the form $\lambda x, y. f(Z(\overline{s_m}))$, where $Z\in \mathcal{FV}$ of type 
$\overline{\gamma_m}\rightarrow \alpha$,  $m>0$,  and for all $1\leq i\leq m$, $s_i\in \mathcal{L}$ and of type $\gamma_i\in\mathcal{T}$.

\begin{lemma}
\label{lem:canonForm}
Let $\mathcal{L}\in\{\Lambda,\Lambda_{sp}\}$ and $C\subseteq \mathcal{G}_{\mathcal{L}}(s,t)$ be complete. Then there exists $g\in C$ such that $g$ is  $\mathcal{L}$-pattern-derived. 
\end{lemma}
\begin{proof}

By the definition of a complete set of generalizations, $C$ must contain a generalization that is at least as specific as  $\lambda x:\alpha.\lambda y:\alpha. f(Z(x,y))$. Any generalization that is at least as specific as  $\lambda x:\alpha.\lambda y:\alpha. f(Z(x,y))$ is referred to as  $\mathcal{L}$-pattern-derived. Thus, $C$ must contain such a generalization. 

\end{proof}

Restricting ourselves to $\mathcal{L}$-pattern-derived generalizations greatly simplifies the structure of the generalizations we must consider. Remember, we aim to show that any ``minimal'' complete set of generalizations contains comparable generalizations. Thus, showing this property for $\mathcal{L}$-pattern-derived generalizations is enough. However, $\mathcal{L}$-pattern-derived generalizations may contain free variables that do not contribute to generalizing the structure of $s$ and $t$. Such variables may be removed through substitution. The resulting generalizations are called $\mathcal{L}$-tight.

\begin{definition}
\label{def:patternDerived}
Let $g$ be an $\mathcal{L}$-pattern-derived generalization of $s\triangleq_{\mathcal{L}} t$. Then $g$ is \textit{$\mathcal{L}$-tight} if for all $W\in  \mathcal{FV}(g)$:
\begin{itemize}
    \item[1)] $g\{W\mapsto \lambda \overline{b_k}. b_i\}\not \in  \mathcal{G}_{\mathcal{L}}(s,t)$, if $W$ has type $\overline{\gamma_k}\rightarrow \gamma_i$ and for $1\leq i \leq k$ and $\gamma_i\in \mathcal{T}$, and
    \item[2)] For any $(\sigma_1,\sigma_2)\in\mathcal{GS}_{gnd}(s,t,g)$,  $g\{W\mapsto t_1\}, g\{W\mapsto t_2\}\not \in  \mathcal{G}_{\mathcal{L}}(s,t)$  where $t_1 = W\sigma_1$, $t_2 = W\sigma_2$.
\end{itemize}
\end{definition}
Essentially, $\mathcal{L}$-tight generalizations are $\mathcal{L}$-pattern-derived generalizations that cannot be made more specific by substituting free variables by \textit{projections} or by the terms in the range substitution pairs contained in $\mathcal{GS}_{gnd}(s,t,g)$. Performing either one of these actions will result in a term that does not generalize $s\triangleq_{\mathcal{L}}t$. 
\begin{example}
The following generalizations of $s\triangleq_{\mathcal{L}} t$ are \textit{$\mathcal{L}$-tight}: 
\begin{itemize}
    \item $\lambda x.\lambda y.f(Z(x,y))$, $\lambda x.\lambda y.f(Z(Z(x,y),Z(x,y)))$
    \item $\lambda x.\lambda y.f(Z(W(x,Z(f(a),a)),W(Z(a,f(a)),y)))$
\end{itemize}
For example, consider the following substitutions: 
\begin{align*}
    \lambda x.\lambda y.f(Z(W(x,Z(f(a),a)),W(Z(a,f(a)),y)))\{Z\mapsto \lambda w_1\lambda w_2. w_1\} &= \lambda x.\lambda y.f(W(x,f(a)))\\
    \lambda x.\lambda y.f(Z(W(x,Z(f(a),a)),W(Z(a,f(a)),y)))\{Z\mapsto \lambda w_1\lambda w_2. w_2\} &= \lambda x.\lambda y.f(y)\\
    \lambda x.\lambda y.f(Z(W(x,Z(f(a),a)),W(Z(a,f(a)),y)))\{W\mapsto \lambda w_1\lambda w_2. w_1\} &= \lambda x.\lambda y.f(Z(x,Z(a,f(a))))\\
    \lambda x.\lambda y.f(Z(W(x,Z(f(a),a)),W(Z(a,f(a)),y)))\{W\mapsto \lambda w_1\lambda w_2. w_2\} &= \lambda x.\lambda y.f(Z(Z(f(a),a),y))
\end{align*}
Notice that none of the righthand side terms is a generalization of $s\triangleq_{\mathcal{L}} t$. On the contrary,  $\lambda x.\lambda y.f(Z(R(x),R(y)))$ and  $\lambda x.\lambda y.f(Z(Z(x,y),Y(x,y)))$ are not \textit{$\mathcal{L}$-tight} since we can substitute $\lambda b_1. b_1$ for $R$ and $\lambda \overline{b_2}. b_2$ for $Y$ respectively.
\end{example}
To simplify the transformation from $\mathcal{L}$-pattern-derived generalizations into \textit{$\mathcal{L}$-tight} generalizations, we only consider ground generalizing substitutions, i.e. substitution pairs in $\mathcal{GS}_{gnd}(s,t,g)$; this is to avoid the introduction of fresh free variables that violate the $\mathcal{L}$-tight conditions while removing violating free variables. Our assumption that $\Sigma$ contains at least one constant of each type in $\mathcal{T}$ allows us to restrict ourselves to ground substitutions. For any substitution whose range contains a free variable $X$ of type $\tau$, we can find a corresponding substitution where $X$ is replaced by a constant $c\in \Sigma$ of type $\tau$. From now on, we assume that the terms comprising the range of generalizing substitutions do not contain free variables. 
\begin{lemma}
\label{lem:minimal}
Let $g$ be an $\mathcal{L}$-pattern-derived generalization of $s\triangleq_{\mathcal{L}} t$. Then there exists an $\mathcal{L}$-tight generalization $g'$ of $s\triangleq_{\mathcal{L}} t$ such that $g\leq_{\mathcal{L}} g'$.
\end{lemma}
\begin{proof}
We proceed by showing how to construct a generalization that does not violate either of the properties listed in  Definition~\ref{def:patternDerived}.
\begin{flushleft}
    \underline{Definition~\ref{def:patternDerived}, case (1)}
\end{flushleft}
Let $W\in \mathcal{FV}(g)$ with type $\overline{\gamma_m}\rightarrow \beta$ be such that for any $1\leq i\leq m$, $g'= g\{W\mapsto \lambda \overline{b_m}. b_i\}\in \mathcal{G}_{\mathcal{L}}(s,t)$. Note that $g'$ has fewer free variables than $g$. Given that $\mathcal{FV}(g)$ is of finite size, after a finite number of applications of such substitutions, the resulting generalizations will no longer violate case (1) of Definition~\ref{def:patternDerived}. 
\begin{flushleft}
    \underline{Definition~\ref{def:patternDerived}, case (2)}
\end{flushleft}
Let $(\sigma_1,\sigma_2)\in\mathcal{GS}_{gnd}(s,t,g)$ and $W\in \mathcal{FV}(g)$ such that,  w.l.o.g,  $g'= g\{W\mapsto W\sigma_1\}\in \mathcal{G}_{\mathcal{L}}(s,t)$. Note that $\sigma_1$ is ground, thus $\mathcal{FV}(g')\subset \mathcal{FV}(g)$. Furthermore,  $\mathcal{FV}(g)$ is of finite size; thus, after a finite number of applications of such substitutions, the resulting generalizations will no longer violate case (2) of Definition~\ref{def:patternDerived}.
\end{proof}
\begin{example}
For $\lambda x.\lambda y.f(Z(R(x),R(y)))$,  both generalizing substitutions map $R$ to the identity function $\lambda w.w$. Note that
$\lambda x.\lambda y.f(Z(R(x),R(y)))\{R\mapsto \lambda w.w\}=\lambda x.\lambda y.f(Z(x,y))$, an \textit{$\mathcal{L}$-tight} generalization of $s\triangleq_{\mathcal{L}} t$.
\end{example}
\begin{observation}
\label{cor:remainpd}
If $g$ is an  $\mathcal{L}$-pattern-derived generalization of $s\triangleq_{\mathcal{L}} t$ and $g'$ is an $\mathcal{L}$-tight generalization of $s\triangleq_{\mathcal{L}} t$ such that $g\leq_{\mathcal{L}} g'$, then $g'$ is $\mathcal{L}$-pattern-derived.
\end{observation}

$\mathcal{L}$-tight generalizations represent the most specific generalization we can construct by removing variables that do not directly generalize the structure of $s$ and $t$. While these generalizations remain $\mathcal{L}$-pattern-derived, they structurally differ from the pattern generalization of $s$ and $t$, that is $\lambda x.\lambda y.f(Z(x,y))$. For the pattern generalization, we derive $s$ and $t$ by substituting the left and right projections into the argument of $f$; this may not be the case for   $\mathcal{L}$-tight generalizations; that is, more complex substitutions may be required. We introduce \textit{pseudo-patterns generalizations}, that is, generalizations more specific than   $\mathcal{L}$-tight generalizations which emulate the structure of the pattern generalization, i.e., the argument of $f$ is a variable which generalizes $s$ and $t$ through the left and right projections. This transformation is performed separately for $\Lambda$ and $\Lambda_{sp}$.  
\subsection{Pseudo-patterns over Simply-typed Lambda Terms}
\label{lam}
First, we consider the transformation for $\Lambda$ and then generalize it to $\Lambda_{sp}$.
\begin{definition}
\label{def:gpseudopattern}
Let $g=\lambda x.\lambda y. f(Z(\overline{s_m}))$ be an $\Lambda$-tight generalization of $s\triangleq_{\Lambda} t$ where $Z$ has type $\overline{\delta_m}\rightarrow \alpha$  for $1\leq i\leq m$, and $s_i$ has type $\delta_i$. Furthermore, let  $(\sigma_1,\sigma_2)\in \mathcal{GS}_{gnd}(s,t,g)$ be substitutions such that $Z\sigma_1=r_1$ and $Z\sigma_2=r_2$, $r_1$ and $r_2$ are of type $\overline{\delta_m}\rightarrow \alpha$, and $Y$ is a variable such that $Y\not\in \mathcal{FV}(g)$ and has type $\alpha\rightarrow\alpha\rightarrow\alpha$. Then the \textit{$(g,\Lambda)$-pseudo-pattern}, denoted $\mathit{G}_{\Lambda}(g,Z,Y,\sigma_1,\sigma_2)$, is
$g\{Z\mapsto \lambda \overline{b_m}.Y(r_1(\overline{b_m}),r_2(\overline{b_m})))\} = \lambda x.\lambda y. f(Y(r_1(\overline{q_m}),r_2(\overline{q_m}))))$, where for all $ 1\leq i\leq m$, $q_i = s_i\{Z\mapsto \lambda \overline{b_m}.Y(r_1(\overline{b_m}),r_2(\overline{b_m})))\} $.
\end{definition}
To simplify the proof outlined in Section~\ref{sec:nullary} we will assume that the  \textit{$(g,\Lambda)$-pseudo-pattern}  $\mathit{G}_{\Lambda}(g,Z,Y,\sigma_1,\sigma_2)$ is a term of the form $\lambda x, y. f(Y(r_1,r_2))$ where $r_1$ and $r_2$ are of type $\alpha$ and $Y$ of type $\alpha\rightarrow\alpha\rightarrow\alpha$.
\begin{example}
Consider the $\Lambda$-tight generalization $ g= \lambda x.\lambda y. f(Z(\lambda w.x,\lambda w.y))$ where $w$ has type $\alpha\rightarrow\alpha$ and $(\{ Z\mapsto \lambda w_1.\lambda w_2. w_1f\},\{ Z\mapsto \lambda w_1.\lambda w_2. w_2f\})\in \mathcal{GS}_{gnd}(s,t,g)$ where $f$ is the constant with type $\alpha\rightarrow\alpha$. Then the $(g,\Lambda)$-pseudo-pattern is $ \lambda x.\lambda y. f(Y(x,y))$.
\end{example}
Note, for any pseudo-pattern $\mathit{G}_{\Lambda}(g,Z,Y,\sigma_1,\sigma_2)$,  $(\{Y\mapsto\lambda x, y. x\}, \{Y\mapsto\lambda x, y. y\})\in \mathcal{GS}_{gnd}(s,t,g)$, thus mimicking the pattern generalization  $\lambda x, y.f(Z(x,y))$. We now show that $\mathit{G}_{\Lambda}(g,Z,Y,\sigma_1,\sigma_2)$ is a generalization of $s\triangleq_{\Lambda} t$.
\begin{lemma}
\label{lem:pseudopattern}
Let $g = \lambda x, y. f(Z(\overline{s_m}))$ be a $\Lambda$-tight generalization of $s\triangleq_{\Lambda} t$, $(\sigma_1,\sigma_2)\in\mathcal{GS}_{gnd}(s,t,g)$, and $g' =  \lambda x, y. f(Y(r_1(\overline{q_m}),r_2(\overline{q_m}))))$, the $\Lambda$-pseudo pattern $\mathit{G}_{\Lambda}(g,Z,Y,\sigma_1,\sigma_2)$. Then $g'$ is a generalization of  $s\triangleq_{\Lambda} t$.

\end{lemma}
\begin{proof}
Let $\sigma_1 =  \sigma_1'\cup \{Z\mapsto r_1\}$ and $\sigma_2 =  \sigma_2'\cup \{Z\mapsto r_2\}$ such that $Z\not \in \dom(\sigma_1')\cup\dom(\sigma_2')$. Note that $g'$ is a partial application of the substitutions $\sigma_1$ and $\sigma_2$ which has been delayed. If we project to the first position of $Y$ and apply $\sigma_1'$, the resulting term is $s$. The computation is as follows:  $g'\{Y\mapsto \lambda w_1\lambda w_2. w_1\} = \lambda x.\lambda y. f(Y(r_1(\overline{q_m}),r_2(\overline{q_m}))))\{Y\mapsto \lambda w_1\lambda w_2. w_1\}$ $=_{\beta}  \lambda x.\lambda y. f(r_1(\overline{p_m}))$ where for all $1\leq i\leq m$, $p_i = q_i\{Y\mapsto \lambda w_1\lambda w_2. w_1\}$. From $\lambda x.\lambda y. f(r_1(\overline{p_m}))$ we derive $s$ by applying $\sigma_1'$. If we start from $g$, we derive the same computation by choosing the $\sigma_1$, which contains the left projection into $Y$, i.e., $\{Z\mapsto r_1\}$.

Essentially, chosing whether to apply $\{Z\mapsto r_1\}$ or $\{Z\mapsto r_2\}$ amounts to substituting $Y$ by the left or right projection. Thus, $\lambda x.\lambda y. f(Z(\overline{s_m}))\{Z\mapsto r_1\}=_{\beta}  \lambda x.\lambda y. f(r_1(\overline{p_m}))$ (applying left projection). In the latter case, applying $\sigma_1'$ results in $s$, as expected. A similar construction can be performed for the right projection. Thus, $g'$ is a generalization of $s\triangleq_{\Lambda} t$ such that  $ (\sigma_1'\cup \{Y\mapsto \lambda w_1\lambda w_2. w_1\},\sigma_2'\cup \{Y\mapsto \lambda w_1\lambda w_2. w_2\})\in \mathcal{GS}_{gnd}(s,t,g
')$. 
\end{proof}

In the next section we show how to \textit{lift} $(g,\Lambda)$-pseudo-patterns  to  $(g,\Lambda_{sp})$-pseudo-patterns. Importantly, $(g,\Lambda)$-pseudo-pattern always have the form $\lambda x.\lambda y. f(Y(r_1,r_2))$ as mentioned above. Similar holds for $(g,\Lambda_{sp})$-pseudo-patterns.

\subsection{Pseudo-patterns over the Super-pattern Fragment}
\label{pat}
The construction presented in Subsection~\ref{lam} may result in a term containing subterms that violate the definition of $\Lambda_{sp}$. This section shows how to lift the generalizations constructed in Subsection~\ref{lam} to $\Lambda_{sp}$. 
\begin{definition}
\label{def:sigmamax}
Let $\lambda\overline{w_m}.r\in\Lambda$. We refer to $q\in \mathit{pos}(r)$ as \textit{$(\Sigma,\overline{w_m})$-representative} if $\mathit{head}(r\vert_{q}\downarrow_{\eta})$ is in $\Sigma$ or an abstraction $\lambda w.$ We refer to  $q\in \mathit{pos}(r)$ as \textit{$(\Sigma,\overline{w_m})$-maximal} if for all $p\in \mathit{pos}(r)$, such that $p<q$, $p$ is not \textit{$(\Sigma,\overline{w_m})$-representative}. The set of \textit{$(\Sigma,\overline{w_m})$-maximal} positions of $r$ is denoted by $m(r,\Sigma,\overline{w_m})$.

\end{definition}
\begin{example}
\label{ex:maximal}
Consider the $\Lambda$-tight generalization $ r= \lambda x.\lambda y. f(Z(\lambda w.x,x,\lambda w_1\lambda w_2.y,f))$. Then the position $1.1.1.1$, $1.1.1.3$,$1.1.1.3.1$, and $1.1.1.4$ are $(\Sigma,\{x,y\})$-representative, while $1.1.1.2$ and $1.1.1$ are not. Additionally $1.1.1.1$, $1.1.1.3$, and $1.1.1.4$ are $(\Sigma,\{x,y\})$-maximal, and  $m(r,\Sigma,\{x,y\})= \{1.1.1.1, 1.1.1.3, 1.1.1.4\}$.
\end{example}
One should consider the maximal positions of a given $\lambda$-term as the positions that can be replaced by a free variable whose arguments are bound variables when in $\eta$-normal form. Replacing maximal terms by free variables is referred to as \textit{lifting}.
\begin{definition}
\label{def:lifting}
Let $\lambda\overline{w_m}.f'(\overline{r_k})\in\Lambda$ such that $f'\in\Sigma$ of type $\overline{\delta_k}\rightarrow\delta_{k+1}$, $w_i$ has type $\gamma_i$, for $1\leq i\leq m$. Furthermore, for all  $1\leq i<j\leq k$, let $H_{q_1}^{r_i}$ and $H_{q_2}^{r_j}$ be pairwise distinct free variables fresh in $\lambda\overline{w_m}.f'(\overline{r_k})$ of type $\overline{\gamma_m}\rightarrow \nu_{q_1}$ and $\overline{\gamma_m}\rightarrow \nu_{q_2}$ respectively, where $q_1\in m(r_i,\Sigma,\overline{w_m})$, and $q_2\in m(r_j,\Sigma,\overline{w_m})$, $\nu_{q_1}$ is the type of $r_i\vert_{q_1}$, and $\nu_{q_2}$ is the type of $r_j\vert_{q_2}$. The lifting of $\lambda\overline{w_m}.f'(\overline{r_k})$, denoted by $\mathit{lift}(\lambda\overline{w_m}.f'(\overline{r_k}),\Sigma) $ is the term $\lambda\overline{w_m}.f'(\overline{r'_k})$ where for each $1\leq i\leq k$ and $q\in m(r_i,\Sigma,\overline{w_m})$, $r_i'\vert_q = H_q^{r_i}(\overline{w_m})$.
\end{definition}
\begin{example}
\label{ex:lift}
Consider the term from Example~\ref{ex:maximal}. Then $$\mathit{lift}(r,\Sigma) = \lambda x.\lambda y. f(Z(H_1(x,y),x,H_2(x,y),H_3(x,y)),$$ where 
$H_1$ and $H_3$ have type $\alpha\rightarrow\alpha\rightarrow ((\alpha\rightarrow\alpha)\rightarrow\alpha)$, and $H_2$ has type $\alpha\rightarrow\alpha\rightarrow ((\alpha\rightarrow\alpha)\rightarrow(\alpha\rightarrow\alpha)\rightarrow\alpha)$.
\end{example}

When constructing a \textit{$(g,\Lambda_{sp})$-pseudo-pattern} from the lifting of a  $(g,\Lambda)$-pseudo-pattern, we only need to consider the term $f$ is applied to, or more precisely the term at position $1.1.1$ of the lifted   $(g,\Lambda)$-pseudo-pattern. Given an $\Lambda$-tight generalization $g$ and its $(g,\Lambda)$-pseudo-pattern $g'$ we refer to the term at position $1.1.1$ as the \textit{core of $g'$} defined as $\mathit{core}(g') =  \mathit{lift}(g',\Sigma)\vert_{1.1.1}$. Additionally, observe that $\lambda x.\lambda y.f(Y(x,y))\leq_{\mathcal{L}}  \mathit{lift}(g',\Sigma)$ as  $\mathit{lift}(g',\Sigma)$ is essentially $g'$ with some subterms replaced by variable expressions.  Thus, we can define the \textit{$(g,\Lambda_{sp})$-pseudo-pattern} through a substitution into $Y$ as is the case in the following definition. 
\begin{definition}
\label{def:gpseudopatternsp}
Let $g$ be an $\Lambda$-tight generalization of $s\triangleq_{\Lambda} t$, and $g'= \lambda x.\lambda y.f(Y(s_1,s_2)$ the \textit{$(g,\Lambda)$-pseudo-pattern} $\mathit{G}_{\Lambda}(g,Z,Y,\sigma_1,\sigma_2)$ (see end of previous section).  Then the \textit{$(g,\Lambda_{sp})$-pseudo-pattern}, denoted by $\mathit{G}_{\Lambda_{sp}}(g,Z,Y,\sigma_1,\sigma_2)$ is $ \lambda x.\lambda y.f(R(x,y))\{R\mapsto \lambda x.\lambda y.\mathit{core}(g')\}$ where $R\not\in \mathcal{FV}(\mathit{lift}(g',\Sigma))$.
\end{definition}

\begin{example}
\label{ex:ppforsp}
Consider the $\Lambda$-tight generalization $g$: $$\lambda x.\lambda y.f(Z(W(x,Z(f(a),a)),W(Z(a,f(a)),y))).$$
A possible pair of substitutions from $\mathcal{GS}_{gnd}(s,t,g)$ is  
$\sigma_1 = \{Z\mapsto \lambda x.\lambda y.x, W\mapsto \lambda x.\lambda y.x\}$ and $\sigma_2 =\{Z\mapsto \lambda x.\lambda y.y, W\mapsto \lambda x.\lambda y.y\}.$
The  \textit{$(g,\Lambda)$-pseudo-pattern} $\mathit{G}_{\Lambda}(g,Z,Y,\sigma_1,\sigma_2)$ is as follows: $g'=\lambda x.\lambda y.f(Y(W(x,f(a)),W(f(a),y))).$
Thus, one possible \textit{$(g,\Lambda_{sp})$-pseudo-pattern}, is as follows: $g^*=\lambda x.\lambda y.f(Y(W(x,H_1(x,y)),W(H_2(x,y),y))),$
where 
$$\sigma_1 = \{Z\mapsto \lambda x.\lambda y.x, W\mapsto \lambda x.\lambda y.x, H_1\mapsto \lambda x.\lambda y.f(a), H_2\mapsto \lambda x.\lambda y.f(a)\}$$
$$\sigma_2 =\{Z\mapsto \lambda x.\lambda y.y, W\mapsto \lambda x.\lambda y.y, H_1\mapsto \lambda x.\lambda y.f(a), H_2\mapsto \lambda x.\lambda y.f(a)\}$$
are substitutions such that $(\sigma_1,\sigma_2)\in\mathcal{GS}_{gnd}(s,t,g^*)$. 

\end{example}

\begin{lemma}
\label{lem:gpseudopatternsp}
Let $g$ be an $\Lambda$-tight generalization of  $s\triangleq_{\Lambda_{sp}} t$. Then every $(g,\Lambda_{sp})$-pseudo-pattern $g'$ is a generalization of $s\triangleq_{\Lambda_{sp}} t$.  
\end{lemma}
\begin{proof}
Similar to the proof of Lemma~\ref{lem:pseudopattern} except we have to extend the generalizing substitutions used to construct the $(g,\Lambda)$-pseudo-pattern associated with $g'$ to reintroduce the terms replaced by variables during the lifting process. 
\end{proof}
Note  $\mathit{G}_{\Lambda_{sp}}(g,Z,Y,\sigma_1,\sigma_2)\leq_{\Lambda} \mathit{G}_{\Lambda}(g,Z,Y,\sigma_1,\sigma_2)$, but it may be the case that  $\mathit{G}_{\Lambda}(g,Z,Y,\sigma_1,\sigma_2)\not \in\mathcal{G}_{\Lambda_{sp}}(s,t)$.
\subsection{Nullarity}
\label{sec:nullary}

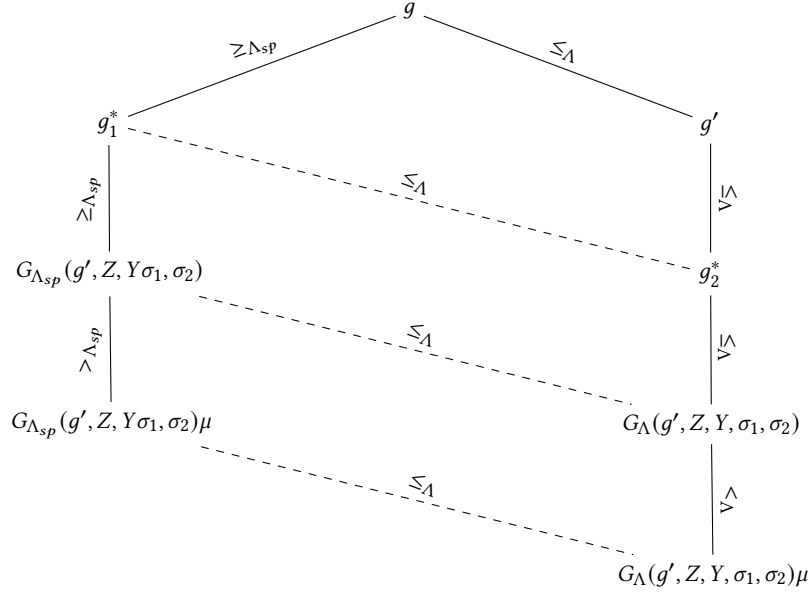
\begin{figure}
\centering
\begin{tikzpicture}
[
level 1/.style = {black, sibling distance = 8cm},
level 2/.style = {black, level distance = 2cm},
level 4/.style = {black, sibling distance = 3cm},
]
\node {$g$} 
child [growth parent anchor = east]  {node (l1) {$g^*_1$} 
child [growth parent anchor = east]  {node (l3) [xshift=-7pt] {$\mathit{G}_{\Lambda_{sp}}(g',Z,Y\sigma_1,\sigma_2)$}
child [growth parent anchor = east] {node (l5) [left,xshift=3pt] {$\mathit{G}_{\Lambda_{sp}}(g',Z,Y\sigma_1,\sigma_2)\mu$}
edge from parent node [right,rotate=90,midway,above] {$>_{\Lambda_{sp}}$}
}
edge from parent node [right,rotate=90,midway,above] {$\geq_{\Lambda_{sp}}$}
}
edge from parent node [right,rotate=20,midway,above] {$\geq_{\Lambda_{sp}}$}}
child {node {$g'$}
child [growth parent anchor = east] {node (l2) [] {$g^*_2$} 
child [growth parent anchor = east]  {node (l4)[xshift=-7pt] {$\mathit{G}_{\Lambda}(g',Z,Y,\sigma_1,\sigma_2)$}
child [growth parent anchor = east] {node (l6) [left,xshift=3pt] {$\mathit{G}_{\Lambda}(g',Z,Y,\sigma_1,\sigma_2)\mu$}
edge from parent node [right,rotate=-90,midway,above] {$<_{\Lambda}$}
}
edge from parent node [right, rotate=-90,xshift=-10pt,yshift=6pt] {$\leq_{\Lambda}$}
}
edge from parent node [right, midway,above, rotate=-90] {$\leq_{\Lambda}$}}
edge from parent node [right, midway,above, rotate=-20] {$\leq_{\Lambda}$}};
\draw[dashed](l1)--(l2) node[midway,above, rotate=-20]{$\leq_{\Lambda}$};
\draw[dashed](l3)--(l4) node[midway,above, rotate=-20]{$\leq_{\Lambda}$};
\draw[dashed](l5)--(l6) node[midway,above, rotate=-20]{$\leq_{\Lambda}$};
\end{tikzpicture}
\caption {Where $g$ is pattern-derived , $g'$ is $\Lambda$-tight, $g^*_1$ is a  $(g, \Lambda_{sp}) \text{-pseudopattern}$ lifted from  $g^*_2$, a  $(g, \Lambda)$ \text{-pseudopattern}, and $\mu = \{Y\mapsto \lambda w_1.\lambda w_2. Y(Y(w_1,w_2),Y(w_1,w_2))\}$.}
\end{figure}
Observe that pseudo-pattern generalizations have a simplified structure compared to $\mathcal{L}$-tight generalizations. In particular, $f$ is applied to a free variable that takes two arguments and there exists $(\sigma_1,\sigma_2)\in \mathcal{GS}_{gnd}(s,t,g)$ such that  $Y\sigma_1 =\lambda x\lambda y. x$ and $ Y\sigma_2=\lambda x\lambda y. y$. Our nullarity argument rests on the following observation concerning the pattern generalization of $s$ and $t$: 

\begin{lemma}
\label{lem:pattstrictorder}
$\lambda x.\lambda y.f(R(x,y))<_{\mathcal{L}} \lambda x.\lambda y.f(R(R(x,y),R(x,y)))$ where $\mathcal{L}\in \{\Lambda,\Lambda_{sp}\}$.
\end{lemma}
\begin{proof}
 We show that there does not exists a subsitution $\sigma$ such that $\lambda x.\lambda y.f(R(R(x,y),R(x,y)))\sigma$ $= \lambda x.\lambda y.f(R(x,y))$. By construction, $R$ has type $\alpha\rightarrow\alpha\rightarrow\alpha$ and $f$ has type $\alpha\rightarrow\alpha$, thus the only possible substitutions into $R$ are 
\begin{itemize}
    \item[1)]$\sigma = \{R\mapsto \lambda w_1\lambda w_2. w_1\}$ or $\sigma = \{R\mapsto \lambda w_1\lambda w_2. w_2\}$,
    \item[2)]$\sigma = \{R\mapsto \lambda w_1\lambda w_2. f'(t_1,\cdots,t_k)\}$ where $f'\in \Sigma$ with type $\overline{\gamma_k}\rightarrow \alpha$, $t_i$ has type $\gamma_i$ for $1\leq i\leq k$ and may contain $w_1$ and/or $w_2$, and $k\geq 0$, or
    \item[3)]$\sigma = \{R\mapsto \lambda w_1\lambda w_2. R'(t_1,\cdots, t_k)\}$ Where $R'$ is free-variable of type $\overline{\delta_{k}}\rightarrow \alpha$ and $t_j$ has type $\delta_j$ and may  contain $w_1$ and/or $w_2$, for $1\leq j\leq k$.
\end{itemize}
Case (1): $\lambda x.\lambda y. f(R(R(x,y),R(x,y)))\sigma \in \{\lambda x.\lambda y. f(x),\ \lambda x.\lambda y. 
f(y)\}$ neither of which are $=_{\alpha\beta\eta}$ to $\lambda x.\lambda y. f(R(x,y))$. Case (2):  $\lambda 
x.\lambda y. f(R(R(x,y),R(x,y)))\sigma = \lambda x.\lambda y. f(f'(t_1',\cdots,t_k'))$ which is not a generalization 
of $s\triangleq_{\mathcal{L}} t$. Case (3): $\lambda x.\lambda y. f(R(R(x,y),R(x,y)))\sigma=  t^*$ such that 
$$\mathit{occ}(R,\lambda x.\lambda y.f(R(R(x,y),R(x,y))))\leq \mathit{occ}(R',t^*).$$ Thus, there does not exist $\sigma$ such that $\lambda x.\lambda y. f(R(R(x,y),R(x,y)))\sigma$ $=_{\alpha\beta\eta}\lambda x.\lambda y. f(R(x,y)) $. 
\end{proof}

Notice that Lemma~\ref{lem:pattstrictorder} does not depend on the arguments of the leaf occurrences of $R$ being $x$ and $y$. We can replace the occurrences of $x$ and $y$ by terms of type $\alpha$, and the relation still holds. For example, within $\lambda x.\lambda y. f(Y(r_1,r_2))$, a $(g,\mathcal{L})$-pseudo-pattern where $\mathcal{L}\in \{ \Lambda, \Lambda_{sp}\}$,  we  replaced $x$ and $y$ by $r_1$ and $r_2$, respectively; this implies the following:

\begin{corollary}
\label{cor:pseudopattstrictorder}
Let $g=\lambda x.\lambda y. f(Z(\overline{s_m}))$ be an $\mathcal{L}$-tight generalization of $s\triangleq_{\mathcal{L}} t$, where $\mathcal{L}\in \{ \Lambda, \Lambda_{sp}\}$, $Y$ a free variable such that $Y\not \in \mathcal{FV}(g)$, and $(\sigma_1,\sigma_2)\in \mathcal{GS}_{gnd}(s,t,g)$. Then 
$$\mathit{G}_{\mathcal{L}}(g,Z,Y,\sigma_1,\sigma_2)<_{\mathcal{L}}\mathit{G}_{\mathcal{L}}(g,Z,Y,\sigma_1,\sigma_2)\{Y\mapsto \lambda w_1.\lambda w_2. Y(Y(w_1,w_2),Y(w_1,w_2))\}.$$ 
\end{corollary}

\noindent Now, we are ready to prove the nullarity of anti-unification over deep $\lambda$-terms.
\begin{theorem}
\label{nullaryOne}
For $\mathcal{L}\in \{\Lambda,\Lambda_{sp}\}$, $\mathcal{L}$-anti-unification is \textit{nullary}.
\end{theorem}
\begin{proof}
Let us assume that  $C\subseteq \mathcal{G}_{\mathcal{L}}(s,t)$ is minimal and complete. By Lemma~\ref{lem:canonForm}, $C$ contains a pattern-derived generalization $g$. By Lemma~\ref{lem:minimal}, $g$ can be transformed into an $\mathcal{L}$-tight generalization $g'$ that is also $\mathcal{L}$-pattern-derived.  Using Lemma~\ref{lem:pseudopattern} \& \ref{lem:gpseudopatternsp} we can derive a $(g',\mathcal{L})$-pseudo-pattern generalization $g''$ of $g'$. Finally by Corollary~\ref{cor:pseudopattstrictorder}, $g^*= g''\{Y\mapsto \lambda w_1.\lambda w_2. Y(Y(w_1,w_2),Y(w_1,w_2))\}$ is strictly more specific than $g''$. This implies that $g\leq_{\mathcal{L}}g'<_{\mathcal{L}}g^*$, entailing that  $C$ is not minimal.
\end{proof}
\section{conclusion} 
We show that deep AU over $\lambda$-terms is nullary, even under substantial restrictions, and thus completely contrasts earlier results for shallow AU~\cite{DBLP:conf/rta/CernaK19}. Furthermore, our nullarity result is inherited by the more expressive type systems of the \textit{ $\lambda$-Cube}~\cite{DBLP:books/daglib/0032840} such as the \textit{calculus of constructions} over which pattern anti-unification was investigated~\cite{pfenning1991unification}. While this makes deep AU significantly less useful in practice, we have not investigated \textit{linear} variants (each variable only occurs once)  that may allow for constructible minimal complete sets. However, the unitarity of top-maximal shallow AU is unlikely to be preserved.

\vspace{.5em}
\noindent\textbf{Acknowledgements:} We would like to thank Temur Kutsia, Martin Suda, and Chad Brown for comments and suggestions that improved this work. 
\bibliographystyle{ACM-Reference-Format}
\bibliography{ref}

\end{document}